\documentclass[a4paper, conference]{IEEEtran}

\IEEEsettopmargin{t}{30mm}
\IEEEquantizetextheight{c}
\IEEEsettextwidth{14mm}{14mm}
\IEEEsetsidemargin{c}{0mm}

\usepackage[utf8]{inputenc}
\usepackage[T1]{fontenc}
\usepackage{url}
\usepackage{cite}
\usepackage{url}
\usepackage{comment}
\usepackage{multicol}
\usepackage{cuted}
\usepackage[cmex10]{amsmath} 
\usepackage{amsmath,amssymb,amsthm,mathtools}
\usepackage{mathrsfs,bm,color}
\usepackage{bbm}
\usepackage{xcolor}
\usepackage{hyperref}
\hypersetup{
	colorlinks,
	linkcolor={red!50!black},
	citecolor={green!80!black},
	urlcolor={blue!80!black}
}
\usepackage{enumitem}

\newcommand{\bs}{\bm{\sigma}}

\newcommand{\bS}{\bm{\Sigma}}

\newcommand{\E}{\mathbb{E}}

\newcommand{\cS}{\mathcal{S}}

\newcommand{\N}{\mathbb{N}}

\newcommand{\R}{\mathbb{R}}
\newcommand{\Z}{\mathbb{Z}}
\newcommand{\Tr}{\operatorname{Tr}}

\DeclareMathOperator{\pP}{\mathbb{P}}

\DeclareMathOperator{\1}{\mathbbm{1}}

\newcommand{\bi}{\bm{i}}
\newcommand{\bx}{\bm{x}}

\newcommand{\bQ}{\bm{Q}}

\newcommand{\bA}{\bm{S}}
\newcommand{\bO}{\bm{O}}

\newcommand{\bP}{\bm{P}}

\newcommand{\bz}{\bm{z}}

\newcommand{\bX}{\bm{X}}
\newcommand{\bY}{\bm{Y}}
\newcommand{\bZ}{\bm{Z}}
\newcommand{\bB}{\bm{B}}

\DeclareMathOperator{\diag}{\mathrm{diag}}

\DeclareMathOperator{\supp}{\mathrm{supp}}

\DeclareMathOperator*{\argmax}{arg\,max}
\DeclareMathOperator*{\argmin}{arg\,min}

\newcommand{\Var}{\operatorname{Var}}

\theoremstyle{plain}
\newtheorem{theorem}{Theorem}
\newtheorem{lemma}{Lemma}

\newtheorem{proposition}{Proposition}

\theoremstyle{definition}

\theoremstyle{remark}
\newtheorem{remark}{Remark}

\begin{document}
\title{Equivalence of Fixed-Rank and Rank-One Even-Order Symmetric Tensor Factorization}

 \author{%
   \IEEEauthorblockN{Ruba Hussein Morsi\IEEEauthorrefmark{1}
     and Anas A. Rahman\IEEEauthorrefmark{2}}
   \IEEEauthorblockA{\IEEEauthorrefmark{1}%
              The University of Turin, Turin, Italy}
   \IEEEauthorblockA{\IEEEauthorrefmark{2}%
              The University of Hong Kong, Hong Kong\\rubahusseinismail.morsi@unito.it, aarahman@hku.hk}
 }

\maketitle
\thispagestyle{plain}
\pagestyle{plain}

\begin{abstract}
In the recent work of Barbier, Ko, and the second present author on sublinear-rank symmetric matrix factorization [Math. Stat. Learn. 9 (2026), 1--68], a key result is that, in the Bayes-optimal setting, the large-size limit of the free entropy of the finite-rank spiked Wigner model is the same as in the rank-one case when the signal has centered i.i.d. entries. In this paper, we show that this rank-one equivalence result extends to the case of finite-rank, even-order, symmetric tensor factorization. Moreover, we give a natural reformulation of a hypothesis that was stated in the aforementioned work to be necessary for this result. As in the matrix case, we use information-theoretic identities and replica symmetry to reduce a known multi-dimensional variational formula for the limiting free entropy to its one-dimensional analog. The novelty stems from the fact that said formula involves a replica symmetric potential containing Hadamard (entrywise) powers, rather than squares, of the matrix-valued variational parameter, so the eigenvalue-based approach used in the matrix case must be adjusted.
\end{abstract}


\section{Introduction}
Signal denoising problems are fundamental to a vast range of disciplines including statistical inference, machine learning, and signal processing. An ubiquitous approach for problems of this type is to model data as being of the ``signal plus noise'' form, with the earliest example being the spiked Wigner model introduced by Johnstone \cite{johnstone} in 2001. See \cite{BKR26} for a review on studies of this and similar spiked matrix models. An important generalization of these models is the spiked tensor model introduced by Montanari and Richard \cite{tensorPCA} in 2014. Studies of this model are motivated by recent and anticipated advances in data science, with tensors naturally encoding, e.g., data sets indexed by more than two parameters, higher-moment analogs of covariance matrices, and hyperedges in hypergraph matching problems; see \cite{tensorPCA,spikedtensor,MLapp} and references therein.

We consider the \textit{rank-$M$, order-$p$, symmetric, spiked tensor model} with ground truth signal matrix $\bX_0\in\R^{N\times M}$ drawn from some prior distribution $\pP_{X,N,M}$ and data generated as
\begin{equation} \label{eq1}
\bY=\sqrt{\frac{\lambda(p-1)!}{N^{p-1}}}\sum_{j=1}^MX_{0,\cdot j}^{\otimes p}+\bZ,
\end{equation}
where $X_{0,\cdot j}$ is the $j$-th column of $\bX_0$, $\bZ$ is an order-$p$ symmetric noise tensor with independently and identically distributed (i.i.d.) entries $Z_{i_1,\ldots,i_p}$ ($1\le i_1\le i_2\le\cdots\le i_p\le N$) drawn from the standard normal distribution $\mathcal{N}(0,1)$, and $\lambda$ is the signal-to-noise ratio (SNR). In terms of components, this is equivalent to observing, for $\bi:=(i_1,\ldots,i_p)\in\mathcal{I}$ with $\mathcal{I}:=\{\bi\in[N]^p\,:\,i_1\le\cdots\le i_p\}$,
\begin{equation}
Y_{\bi}=\sqrt{\frac{\lambda(p-1)!}{N^{p-1}}}\sum_{j=1}^MX_{0,i_1j}X_{0,i_2j}\cdots X_{0,i_pj}+Z_{\bi}.
\end{equation}

In the Bayes-optimal setting, where the signal reconstruction procedure may take into account knowledge of the form \eqref{eq1} of $\bY$, the prior distribution $\pP_{X,N,M}$, and the parameters $\lambda,M$, the statistical information of the above model is encoded in the \textit{free entropy} (equivalently, up to a simple additive term, the mutual information $I(\bX_0;\bY)$; see, e.g., \cite{LelargeMiolane} and cf.~Section \ref{s5})
\begin{equation} \label{eq3}
F_N(\lambda):=F_{N,M}(\lambda)=\frac{1}{NM}\E_{\bZ,\bX_0}\ln Z_N,
\end{equation}
where
\begin{equation}
Z_{N,M}:=\int_{\R^{N\times M}}e^{H_N(\bX)}\,d\pP_{X,N,M}(\bX)
\end{equation}
is the \textit{partition function} associated with the \textit{Hamiltonian}
\begin{align}
H_N(\bX)&:=\sum_{\bi\in\mathcal{I}}\Big(Y_{\bi}\pi_{\bi}-\frac{1}{2}\pi_{\bi}^2\Big),
\\\pi_{\bi}&:=\sqrt{\frac{\lambda(p-1)!}{N^{p-1}}}\sum_{j=1}^MX_{i_1j}X_{i_2j}\cdots X_{i_pj}.
\end{align}
In the terminology of Bayesian inference, the partition function $Z_{N,M}$ is the normalization of the posterior distribution $\pP_{X|Y,N,M}(\bX\mid \bY(\bX_0,\bZ))$ and the Hamiltonian $H_N(\bX)$ is the log-likelihood.

The large $N$ limit of the free entropy $F_N(\lambda)$ was first specified in the $p=2$ matrix case, with studies progressively addressing the regimes of $M=1$, $M={\rm O}(1)$, and most recently, $M={\rm o}(N^{1/20})$ \cite{reeves} and $M={\rm o}(\sqrt{\ln N})$ \cite{BKR26} (literature reviews on the finite-rank spiked matrix models can be found in these latter references). Works concerning proper spiked tensor models corresponding to $p>2$ have primarily been focused on the $M=1$ regime, with proof methods including the cavity method \cite{spikedtensor} (the heuristic replica method was also used herein to conjecture a finite-$M$ formula), adaptive interpolation \cite{ai1,ai2}, and Hamilton--Jacobi equations \cite{mourrat,mourrat2}; see also \cite{JLM20}. The adaptive interpolation method of \cite{ai1,ai2} was soon after pushed to the finite-rank case in \cite{adaptiveinterpolation}, where it was shown that for $p$ even and $\bX_0$ having i.i.d. rows drawn from a distribution $\pP_{X,M}$ with bounded support,
\begin{equation} \label{eq7}
\lim_{N\to\infty} F_N(\lambda)=\sup_{\bA\in\cS_M}F^{\rm RS}_{M,p}(\bA,\lambda),
\end{equation}
where $\cS_M$ denotes the set of symmetric positive semidefinite $M\times M$ matrices and, letting $\bz\sim\mathcal{N}(0,I_M)$ and $\bx_0\sim\pP_{X,M}$,
\begin{align}
F^{\rm RS}_{M,p}(\bA,\lambda)&:=\frac{1}{M}\E_{\bz,\bx_0}\ln Z^{\rm RS}_{M,p} -\frac{\lambda(p-1)}{2Mp}\sum_{\ell,\ell'=1}^M(\bA^{\circ p})_{\ell\ell'}, \label{eq8}
\\Z^{\rm RS}_{M,p}&:=\int_{\R^M}e^{\sqrt{\lambda}\bx^\intercal\sqrt{\bQ}\bz+\lambda\bx_0^\intercal\bQ\bx-\tfrac{\lambda}{2}\bx^\intercal\bQ\bx}\,d\pP_{X,M}(\bx) \label{eq9}
\end{align}
are the \textit{rank-$M$, order-$p$ replica symmetric potential} and \textit{partition function}. Here, $\bA^{\circ p}=[\bA_{ij}^p]_{i,j=1}^M$ denotes the $p$-th Hadamard power of $\bA$ and we henceforth write $\bQ:=\bA^{\circ(p-1)}$.

In the $p=2$, finite-$M$ case, it was shown in \cite{BKR26} that when the prior distribution is taken to factorize further into the form $\pP_{X,M}=\pP_X^{\otimes M}$ with some additional technical hypotheses on $\pP_X$, the right-hand side of equation \eqref{eq7} actually reduces to its $M=1$ equivalent, thereby demonstrating equivalence between the fixed-rank and rank-one spiked Wigner models considered therein (in fact, this was proven for $M={\rm o}(\sqrt{\ln N})$, but we are interested in the intermediary finite-$M$ result); see also \cite{bates2023parisi,chen2023potts} for analogous simplifications of vector spin Parisi formulas. Our goal in the remainder of this paper is to continue this line of work and show that this phenomenon holds true for all even integers $p$ and finite $M$. Our approach is the same as in \cite{BKR26}, except that arguments involving the eigenvalues of $\bA$ must now be rewritten in terms of its entries; note that when $p=2$, $\sum_{\ell,\ell'=1}^M(\bA^{\circ p})_{\ell\ell'}=\Tr\bA^2$, so the second term in the right-hand side of equation \eqref{eq8} can be easily expressed in terms of the eigenvalues of $\bA$, which moreover coincides with $\bQ$.

\section{Outline and summary of results}
We first revise some basic properties of Hadamard products.
\begin{lemma} \label{Lemma1}
Fix $M,p\in\N$ with $p$ even. Let $\bA\in\cS_M$ and $\bQ=\bA^{\circ(p-1)}$. Then, the following are true:
\begin{enumerate}
\item By the Schur Product Theorem \cite{schur}, the Hadamard product of two symmetric positive (semi-)definite matrices is also symmetric positive (semi-)definite. Thus, we have in particular that $\bQ\in\cS_M$.
\item Define $\psi:\R^{M\times M}\to\R^{M\times M}$, $\bA\mapsto\bA^{\circ(p-1)}$. Then, as $p-1$ is a positive odd integer, $\psi$ is a bijection on its domain $\R^{M\times M}$ in addition to $\cS_M$ and $\psi^{-1}(\cS_M)$.
\item As $p$ is an even integer, the entries of $\bQ\circ\bA=[\bA_{ij}^p]_{i,j=1}^M$ are all non-negative.
\item A simple computation shows that $\bA$ being symmetric implies that $\sum_{\ell,\ell'=1}^M(\bQ\circ\bA)_{\ell\ell'}=\Tr \bQ\bA$. Furthermore, this quantity is non-negative by the previous point.
\end{enumerate}
\end{lemma}
In view of these facts, we opt to reformulate the right-hand side of equation \eqref{eq7} as a supremum over $\psi(\cS_M)$. Indeed, taking this opportunity to further re-express this quantity in terms of the mutual information between the input signal and output data of a vector Gaussian channel (see \cite[Appendix~A]{BKR26}), we observe that
\begin{align}
\sup_{\bA\in\cS_M}F^{\rm RS}_{M,p}(\bA,\lambda)=&\sup_{\bQ\in\psi(\cS_M)}\widehat{F}^{\rm RS}_{M,p}(\bQ,\lambda), \label{eq10}
\\\widehat{F}^{\rm RS}_{M,p}(\bQ,\lambda):=&-\frac{1}{M}I(\bx_0;\sqrt{\lambda\bQ}\bx_0+\bz) \nonumber
\\&+\frac{\lambda}{2M}\rho\Tr\bQ-\frac{\lambda(p-1)}{2Mp}\Tr(\bQ\bA), \label{eq11}
\end{align}
where we define $\rho:=\E_{\pP_X}X^2$ and, for ease of notation, continue to write $\bA$ for the inverse of $\bQ$ under the mapping $\psi$ described in point 2) of Lemma \ref{Lemma1}. Thus, our goal is to prove that, upon adopting some additional hypotheses on $\pP_X$ drawn from \cite{BKR26}, the right-hand side of equation \eqref{eq10} reduces to its rank-one analog, thereby simplifying equation \eqref{eq7}. In fact, due to the complicated structure of $\psi(\cS_M)$, we will instead work to simplify $\sup_{\bQ\in\cS_M}\widehat{F}^{\rm RS}_{M,p}(\bQ,\lambda)$ and then show that the result agrees with the right-hand side of equation \eqref{eq10}. Our main theorem is as follows.
\begin{theorem} \label{thrm1}
Fix $M,p\in\N$, with $p$ even. Following \cite{BKR26}, assume the following hypotheses:
\begin{enumerate}[label=\upshape(H\arabic*),ref=(H\arabic*)]
\item\label{H1} The prior distribution on $\bX_0$ fully factorizes as
\begin{equation*}
\pP_{X,N,M}(\bX_0)=\prod_{i=1}^N\prod_{j=1}^M\pP_X(X_{0,ij}).
\end{equation*}
\item\label{H2} The distribution $\pP_X$ is centered with \textit{$D$-bounded support}, meaning that there exists $0\le D<\infty$ such that
\begin{equation*}
\supp\pP_X\in[-D,D].
\end{equation*}
\item\label{H3} The distribution $\pP_X$ is discrete, absolutely continuous, or a mixture of two such distributions.
\item\label{H4} The distribution $\pP_X$ is such that for $\pP_{X,M}=\pP_X^{\otimes M}$ and all constant $m\le M$,
\begin{equation} \label{phiM}
\phi_m(\lambda):=\sup_{\bA\in\cS_m}F_{m,p}^{\mathrm{RS}}(\bA,\lambda)
\end{equation}
is real analytic in $\lambda$ on $[0,\infty)$ except for possibly one critical point $0<\lambda_c<\infty$, independent of $m$.
\end{enumerate}
Setting $\rho:=\E_{\pP_X}X^2$, the limiting free entropy \eqref{eq3} of the spiked tensor model~\eqref{eq1} is then given in terms of the rank-one replica symmetric potential \eqref{eq8}, \eqref{eq11} by
\begin{align} \label{mainresult}
\lim_{N\to\infty}F_N(\lambda)&=\sup_{s\in[0,\rho]}F_{1,p}^{\mathrm{RS}}(s,\lambda) \nonumber
\\&=\sup_{q\in[0,\rho^{p-1}]}\widehat{F}_{1,p}^{\mathrm{RS}}(q,\lambda).
\end{align}
\end{theorem}

\begin{remark}
Comments on the above hypotheses can be found in \cite{BKR26}. We have replaced a hypothesis on the monotonicity of $\lambda^2\mathrm{mmse}(\lambda)$ at large $\lambda$ with the more natural and sufficient hypothesis \ref{H3}. Here, $\mathrm{mmse}(\lambda)$ is the minimum mean-square error (MMSE) of the scalar Gaussian channel $y=\sqrt{\lambda}x_0+z$ with $x_0\sim\pP_X$ and $z\sim\mathcal{N}(0,1)$. Said hypothesis of \cite{BKR26} is satisfied by the more natural hypothesis \ref{H3}; see the proof of Proposition \ref{prop3} in Section~\ref{s4} below. By the Lebesgue decomposition theorem \cite{cinlar}, every probability distribution is a mixture of a discrete, an absolutely continuous, and a singular continuous (e.g., Cantor) distribution.
\end{remark}

The proof of Theorem \ref{thrm1} follows from interpreting the replica symmetric potential $\widehat{F}^{\rm RS}_{M,p}(\bQ,\lambda)$ as the mutual information term $-\frac{1}{M}I(\bx_0;\sqrt{\lambda\bQ}\bx_0+\bz)$ plus a regularizing tracial term, and then bounding each of these in terms of their scalar equivalents. Bounding of the tracial term is done through standard tools like Jensen's inequality, while the decoupling of the mutual information is due to the following identities proven in \cite{BKR26}.
\begin{lemma} \label{Lemma2}
Assume hypothesis \ref{H1} and let $\bx_0\sim\pP_{X,M}$, $x_0\sim\pP_X$, $\bz\sim\mathcal{N}(0,I_M)$, and $z\sim\mathcal{N}(0,1)$. Then, for $\bS\in\cS_M$, the mutual information between $\bx_0$ and the output of the vector channel with additive Gaussian noise $\bS^{1/2}\bz$ satisfies
\begin{equation} \label{eq14}
I(\bx_0;\bx_0+\bS^{1/2}\bz)\ge\sum_{i=1}^MI(x_0;x_0+\sqrt{\bS_{ii}}\bz).
\end{equation}
If we further assume hypothesis \ref{H2} and that $\bS$ is such that $\bS_{ii}>D^2$ for each $i\in[M]$, then
\begin{equation} \label{eq15}
I(\bx_0;\bx_0+\bS^{1/2}\bz)\ge MI(x_0;x_0+\sqrt{\sigma}z),
\end{equation}
where $\sigma:=\Tr\bS/M$. Inequality \eqref{eq14} is an equality when $\bS$ is diagonal, while inequality \eqref{eq15} is an equality when $\bS=\sigma I_M$.
\end{lemma}
This lemma is used in Section \ref{s4} with $\bS=(\lambda\bQ)^{-1}$ --- it turns out that the competition between the mutual information and tracial terms (each decreases when moving in $\cS_M$ towards the maximizer of the other) is best treated by working with the eigenvalues and diagonal entries of $\bS$ rather than $\bQ$. Even so, said competition remains intractable in general and we must simplify it further by working in the regimes of $\bS$ having all eigenvalues very small or very large. These two regimes correspond respectively to high and low SNR due to the following properties of the maximizers $\bQ^*$ of $\widehat{F}_{M,p}^{\rm RS}(\bQ,\lambda)$ that we prove in Section \ref{s3}.
\begin{proposition} \label{prop1}
Assume hypotheses \ref{H1}, \ref{H2} and let
\begin{align*}
\bQ^*&:=\argmax_{\bQ\in\cS_M}\widehat{F}^{\rm RS}_{M,p}(\bQ,\lambda),\quad\bA^*=\psi^{-1}(\bQ^*).
\end{align*}
so that $\bQ^*$ is either a critical point or lies on the boundary $\partial\cS_M$ of $\cS_M$. Then, $\bQ^*$ and its eigenvalues $q_1^*, \ldots,q_M^*$ have the following properties:
\begin{enumerate}
\item There exists finite $\rho'$ such that $0\le q_1^*,\ldots,q_M^*\le\rho'$ for all $\lambda\ge0$.
\item If $\bQ^*$ is a critical point, $0\le q_1^*,\ldots,q_M^*\le \rho^{p-1}$ for all $\lambda\ge0$.
\item Given $\rho_L'\in(0,\rho^{p-1})$, there exists $\lambda_L'>\lambda_c$ such that for all $\lambda>\lambda_L'$, $\bQ^*\notin\partial\cS_M$ and $\rho_L'<q_1^*,\ldots,q_M^*\le \rho^{p-1}$.
\end{enumerate}
\end{proposition}

Looking carefully at Lemma \ref{Lemma2}, one sees that the first inequality, valid for all $\bS$, is expressed in terms of the diagonal entries of $\bS_{ii}$, while the latter, valid only in the low SNR regime, involves the normalized trace $\sigma$ of $\bS$. We will see in Section \ref{s4} that both are necessary: The first inequality is used when the tracial term in $\widehat{F}^{\rm RS}_{M,p}(\bQ,\lambda)$, rewritten in terms of $\bS$, can be bounded by a scalar function of $\sigma$ --- this is done via Jensen's inequality, which is only available at high SNR, since that is when the tracial term has the necessary convexity properties. The second inequality, on the other hand, requires low SNR but can be used without any further requirements on the tracial term. Interestingly, the proofs corresponding to high and low SNR are dual in that the first deals with the diagonal entries of $\bS$, while the second involves its eigenvalues, instead (note that these agree when $\bS$ is diagonal). This is in addition to the duality concerning whether the mutual information term or the tracial term is to be bounded by a scalar equivalent depending only on $\sigma$.

As we proceed, let us remind that our strategy is largely the same as in \cite{BKR26}, so we focus on the novel aspects arising from needing to treat the Hadamard power $\bQ=\bA^{\circ(p-1)}$.

\section{Properties of maximizers} \label{s3}
We first prove a fixed point equation governing maximizers of $\widehat{F}^{\rm RS}_{M,p}(\bQ,\lambda)$ that are critical points and otherwise characterizing the non-singular maximizers.
\begin{lemma} \label{Lemma3}
Assume hypotheses \ref{H1}, \ref{H2} and let $\rho=\E_{\pP_X}X^2$. Define the Gibbs average associated with $\widehat{F}^{\rm RS}_{M,p}(\bQ,\lambda)$ by
\begin{equation}
\langle\,\cdot\,\rangle_{\rm RS}:=\frac{1}{Z^{\rm RS}_{M,p}}\int\,\cdot\,e^{\sqrt{\lambda}\bx^\intercal\sqrt{\bQ}\bz+\lambda\bx_0^\intercal\bQ\bx-\tfrac{\lambda}{2}\bx^\intercal\bQ\bx}\,d\pP_{X,M}(\bx),
\end{equation}
where $Z^{\rm RS}_{M,p}$ is as in equation \eqref{eq9}. Then, maximizers $\bQ$ of $\widehat{F}^{\rm RS}_{M,p}(\bQ,\lambda)$ either satisfy the fixed point equation
\begin{equation} \label{eq17}
\bA=\psi^{-1}(\bQ)=\E_{\bz,\bx_0}\langle\bx\bx_0^\intercal\rangle_{\rm RS}
\end{equation}
or are boundary points such that
\begin{equation} \label{eq18}
\bA=\E_{\bz,\bx_0}\langle\bx\bx_0^\intercal\rangle_{\rm RS}+\bP
\end{equation}
for some $M\times M$ symmetric positive semidefinite matrix $\bP$.
\end{lemma}
\begin{proof}
The multivariate I-MMSE relation \cite{palomar,reeves2018} states that the gradient of $I(\bx_0;\sqrt{\lambda\bQ}\bx_0+\bz)$ with respect to $\bQ\in\cS_M$ is given by half of the MMSE matrix according to
\begin{equation} \label{eq19}
\nabla_{\bQ}I(\bx_0;\sqrt{\lambda\bQ}\bx_0+\bz)=\frac{\lambda}{2}\left(\rho I_M-\E_{\bz,\bx_0}\langle\bx\bx_0^\intercal\rangle_{\rm RS}\right);
\end{equation}
this relation has unambiguous meaning on $\partial\cS_M$ when it is interpreted in terms of appropriate directional derivatives. Combining this with the fact that
\begin{equation*}
\nabla_{\bQ}\Tr(\bQ\bA)=\nabla_{\bQ}\sum_{\ell,\ell'=1}^M\bQ_{\ell\ell'}^{\frac{p}{p-1}}=\frac{p}{p-1}\bA,
\end{equation*}
we have by equation \eqref{eq11} that
\begin{equation} \label{eq20}
\nabla_{\bQ}\widehat{F}^{\rm RS}_{M,p}(\bQ,\lambda)=\frac{\lambda}{2M}\left(\E_{\bz,\bx_0}\langle\bx\bx_0^\intercal\rangle_{\rm RS}-\bA\right).
\end{equation}
Setting this gradient to zero yields equation \eqref{eq17}.

If $\bQ$ is a maximizer that is not a critical point, it must lie in $\partial\cS_M$ and hence have vanishing eigenvalues (eigenvalues tending to infinity are disqualified due to point 1) of Proposition \ref{prop1}). Denote the eigenvalues of $\bQ$ as $q_1,\ldots,q_M$ with corresponding eigenvectors $\bO_1,\ldots,\bO_M$ and suppose that $q_1,\ldots,q_d=0$ for some $d\in[M]$. Then, applying the chain rule to equation \eqref{eq20} shows that
\begin{equation*}
f_i:=\frac{\partial}{\partial q_i}\widehat{F}^{\rm RS}_{M,p}(\bQ,\lambda)=\frac{\lambda}{2M}\bO_i^\intercal\left(\E_{\bz,\bx_0}\langle\bx\bx_0^\intercal\rangle_{\rm RS}-\bA\right)\bO_i.
\end{equation*}
For our described $\bQ$ to be a non-critical maximizer, we require that these partial derivatives be negative for $i\in[d]$ and zero for $i=d+1,\ldots,M$ at $\bQ$. Thus, letting $\bO$ be the matrix whose columns are the eigenvectors $\bO_1,\ldots,\bO_M$,
\begin{equation*}
\bP:=-\frac{2M}{\lambda}\bO\diag(f_1,\ldots,f_M)\bO^\intercal
\end{equation*}
is a symmetric positive semidefinite matrix. Equation \eqref{eq18} follows from using the chain rule once again to note that the gradient \eqref{eq20} equals $-\lambda\bP/(2M)$.
\end{proof}
The converse does not hold: Solutions of equation \eqref{eq17} are not guaranteed to be maximizers of $\widehat{F}^{\rm RS}_{M,p}(\bQ,\lambda)$, only stationary points. Note also that $\E_{\bz,\bx_0}\langle\bx\bx_0^\intercal\rangle_{\rm RS}$ is positive semidefinite by Bayes-optimality, so maximizers specified by equation \eqref{eq17} are actually such that
\begin{equation*}
\bA\in\cS_M\implies\bQ\in\psi(\cS_M).
\end{equation*}

\begin{remark}
In the matrix case, the eigenvalues of $\bQ$ are squares of those of $\bA$, so it is possible to prove a version of Lemma~\ref{Lemma3} expressed explicitly in terms of the eigenvalues of maximizers $\bQ$. This is Lemma 3.1 of \cite{BKR26}. We note however that the proof given in \cite{BKR26} for this lemma does not produce equation \eqref{eq17} when $\bQ$ lies in the interior of $\cS_M$, as erroneously implied therein --- one must instead appeal to the proof given above. Nonetheless, this result is not used in \cite{BKR26}, so the work is internally consistent and, moreover, the result can be recovered from the main theorem of said work.
\end{remark}
With Lemma \ref{Lemma3} in hand, we now prove Proposition \ref{prop1}.
\begin{lemma} \label{Lemma4}
Let $\bQ^*\in\cS_M$ be a maximizer of $\widehat{F}^{\rm RS}_{M,p}(\bQ,\lambda)$. Then, there exists finite $\rho'$ such that, for all $\lambda\ge0$, the eigenvalues $q_1^*,\ldots,q_M^*$ of $\bQ^*$ are such that
\begin{equation*}
0\le q_1^*,\ldots,q_M^*<\rho'.
\end{equation*}
\end{lemma}
\begin{proof}
We consider the tracial term of equation \eqref{eq11} with $\bQ$ and $\bA$ replaced by $\bQ^*$ and $\bA^*=\psi^{-1}(\bQ^*)$, respectively. By points 3) and 4) of Lemma \ref{Lemma1}, we have that
\begin{equation*}
\Tr(\bQ^*\bA^*)=\sum_{\ell,\ell'=1}^M(\bA^*)_{\ell\ell'}^p\ge\sum_{\ell=1}^M(\bQ^*)_{\ell\ell}^{p/(p-1)}.
\end{equation*}
Then, using Jensen's inequality shows that
\begin{multline} \label{eq21}
\frac{\lambda}{2M}\rho\Tr\bQ^*-\frac{\lambda(p-1)}{2Mp}\Tr(\bQ^*\bA^*)
\\\le\frac{\lambda}{2M}\rho\Tr\bQ^*-\frac{\lambda(p-1)}{2M^{p/(p-1)}p}(\Tr\bQ^*)^{p/(p-1)}.
\end{multline}
Suppose for the sake of contradiction that $\bQ^*$ has an unbounded eigenvalue. Then, since all eigenvalues of $\bQ^*$ are non-negative, we must have that $\Tr\bQ^*\to\infty$. As $(\Tr\bQ^*)^{p/(p-1)}$ dominates $\Tr\bQ^*$ in this limit, the right-hand side of the above inequality must tend to negative infinity. This is then also true for $\widehat{F}^{\rm RS}_{M,p}(\bQ^*,\lambda)$ since mutual information is non-negative. Therefore, such a $\bQ^*$ cannot be a maximizer of $\widehat{F}^{\rm RS}_{M,p}(\bQ,\lambda)$ and we arrive at our sought contradiction. There must hence exist a finite bound $\rho'$ on the eigenvalues of $\bQ^*$. 
\end{proof}

\begin{lemma} \label{Lemma5}
Retain the setting of Lemma \ref{Lemma3} and let $q_1^*,\ldots,q_M^*$ be the eigenvalues of a critical point $\bQ^*\in\cS_M$ of $\widehat{F}^{\rm RS}_{M,p}(\bQ,\lambda)$. Then, for all $\lambda\ge0$,
\begin{equation*}
0\leq q_1^*,\ldots,q_M^*\le\rho^{p-1}.
\end{equation*}
\end{lemma}
\begin{proof}
The lower bound follows trivially since $\bQ^*$ is positive semidefinite. For the upper bound, we first note that eigenvalues of $\E_{\bz,\bx_0}\langle\bx\bx_0^\intercal\rangle_{\rm RS}$ are bounded above by $\rho$ due to Bayes-optimality; see the proof of Lemma 3.2 of \cite{BKR26} with $\bO_i$ therein replaced by eigenvectors of $\E_{\bz,\bx_0}\langle\bx\bx_0^\intercal\rangle_{\rm RS}$. By Lemma \ref{Lemma3}, the eigenvalues of $\bA^*=\psi^{-1}(\bQ^*)$ then lie in $[0,\rho]$. Now, for two symmetric positive semidefinite matrices $\bm{A},\bm{B}$ with maximal eigenvalues $\lambda_{\rm max}(\bm{A}),\lambda_{\rm max}(\bm{B})$, it is known \cite[Ch.~5]{hornjohnson} that the maximal eigenvalue of their Hadamard product obeys the bound
\begin{equation*}
\lambda_{\rm max}(\bm{A}\circ\bm{B})\le \lambda_{\rm max}(\bm{A})\lambda_{\rm max}(\bm{B}).
\end{equation*}
Thus, we have by induction that
\begin{equation*}
\lambda_{\rm max}(\bQ^*)\le\lambda_{\rm max}(\bA^*)^{p-1}\le\rho^{p-1},
\end{equation*}
as required.
\end{proof}

\begin{lemma} \label{Lemma6}
Assume hypotheses \ref{H1}, \ref{H2}, fix $\rho_L'\in(0,\rho^{p-1})$, and let $\bQ^*\in\cS_M$ be a maximizer of $\widehat{F}^{\rm RS}_{M,p}(\bQ,\lambda)$. Then, there exists $\lambda_L'>\lambda_c$ such that whenever $\lambda>\lambda_L'$, $\bQ^*$ is non-singular and its eigenvalues $q_1^*,\ldots,q_M^*$ are such that
\begin{equation*}
\rho_L'<q_1^*,\ldots,q_M^*\le\rho^{p-1}.
\end{equation*}
\end{lemma}
\begin{proof}
Throughout this proof, we treat singular $\bQ^*$ by replacing it with $\bQ_\epsilon^*:=\bQ^*+\epsilon I_M$ with $\epsilon>0$ to be eventually taken to zero. Then, as $\bQ_\epsilon^*$ is non-singular, the MMSE matrix on the right-hand side of the I-MMSE relation \eqref{eq19} vanishes as $\lambda\to\infty$ \cite[Ch.~11,12]{kay}. Moreover, as this matrix is positive semidefinite and a continuous function of SNR \cite{reeves2018}, we may establish the existence of $\lambda_L'>\lambda_c$ such that for all $\lambda>\lambda_L'$,
\begin{equation} \label{eq22}
\E_{\bz,\bx_0}\langle\bx\bx_0^\intercal\rangle_{\rm RS}=\rho I_M-\bP_{\delta},
\end{equation}
where $\bP_{\delta}$ is a symmetric positive semidefinite matrix with all entries bounded in magnitude by some small $\delta>0$ yet to be chosen. For each $\lambda$, we may now take $\epsilon\to0$ (updating the definition of $\delta$ if needed).

We first prove that $\bQ^*$ must be a critical point. Thus, for the sake of contradiction, assume otherwise. Then $\bQ^*$ must be a boundary point with some vanishing eigenvalues and we have by equation \eqref{eq18} that there is a symmetric positive semidefinite matrix $\bP$ such that
\begin{equation*}
\bA^*=\psi^{-1}(\bQ^*)=\E_{\bz,\bx_0}\langle\bx\bx_0^\intercal\rangle_{\rm RS}+\bP.
\end{equation*}
Then, for $\lambda>\lambda_L'$, we have by equation \eqref{eq22} that
\begin{equation*}
\bA^*=\rho I_M-\bP_{\delta}+\bP.
\end{equation*}
Now, for $k\in[p-1]$, $\bP^{\circ k}$ is positive semidefinite and writing $\bB:=\rho I_M-\bP_{\delta}$, we compute that $\bB^{\circ k}$ has diagonal entries bounded below by $\rho^k-k\delta\rho^{k-1}$ and off-diagonal entries bounded above by $\delta^k$. By the Gerschgorin circle theorem \cite{gershgorin} (see \cite{hornjohnson} for a textbook treatment), the eigenvalues of $\bB^{\circ k}$ are then bounded below according to
\begin{equation} \label{eq23}
\lambda_{\rm min}(\bB^{\circ k})\ge \rho^k-k\delta\rho^{k-1}-(M-1)\delta^k,
\end{equation}
which is strictly positive for small enough $\delta$, hence $\bB^{\circ k}$ is positive definite. Next, we consider the binomial expansion
\begin{align*}
\bQ^*&=(\bB+\bP)^{\circ(p-1)}
\\&=\sum_{k=0}^{p-1}\binom{p-1}{k}\bB^{\circ(p-1-k)}\circ\bP^{\circ k}.
\end{align*}
By the Schur product theorem (recall point 1) of Lemma~\ref{Lemma1}), every term in this sum is positive semidefinite, so the eigenvalues of $\bQ^*$ are bounded below by those of $\bB^{\circ(p-1)}$. This is a contradiction with the assumption that $\bQ^*$ has vanishing eigenvalues by way of comparison with the bound \eqref{eq23}. Hence, for $\lambda>\lambda_L'$ (with $\lambda_L'$ chosen large enough so that $\delta$ is small enough for the above arguments to go through), all maximizers $\bQ^*$ must be critical points.

Moving on to the final statement of the lemma, since $\bQ^*$ is a critical point by the above, we have by Lemma \ref{Lemma5} that $q_1^*,\ldots,q_M^*\le\rho^{p-1}$. Moreover, combining equations \eqref{eq17}, \eqref{eq22} shows that
\begin{align*}
\bA^*&=\E_{\bz,\bx_0}\langle\bx\bx_0^\intercal\rangle_{\rm RS}
\\&=\rho I_M-\bP_{\delta}
\end{align*}
for all $\lambda>\lambda_L'$, where $\lambda_L',\delta$ are as above. Observing that this expression is equivalent to the $\bB$ defined above, we may use the bound \eqref{eq23} to see that
\begin{equation*}
\lambda_{\rm min}(\bQ^*)\ge\rho^{p-1}-(p-1)\delta\rho^{p-2}-(M-1)\delta^{p-1}.
\end{equation*}
This concludes the proof if the right-hand side of the above is greater than $\rho_L'$. Otherwise, we may further reduce $\delta$ (and increase $\lambda_L'$) so that this is the case.
\end{proof}

\section{Rank-one reduction} \label{s4}
As mentioned earlier, we now work with $\bS=(\lambda\bQ)^{-1}$. Note that since $\bQ\in\cS_M$, so too is $\bS$. Rewriting $\widehat{F}^{\rm RS}_{M,p}(\bQ,\lambda)$ in terms of $\bS$, we see that maximizing this quantity is equivalent to the problem of minimizing
\begin{multline} \label{eq24}
\widetilde{F}^{\rm RS}_{M,p}(\bS,\lambda):=I(\bx_0;\bx_0+\bS^{1/2}\bz)-\frac{\rho}{2}\Tr\bS^{-1}
\\+\frac{p-1}{2p}\lambda^{-1/(p-1)}\sum_{\ell,\ell'=1}^M(\bS^{-1})^{p/(p-1)}_{\ell\ell'}
\end{multline}
over $\bS\in\cS_M$. It is immediate from Lemma \ref{Lemma2} that when $\bS$ is diagonal, this potential decouples into a sum of $M$ copies of its one-dimensional equivalent. Our goal is to show that $\widetilde{F}^{\rm RS}_{M,p}(\bS,\lambda)$ is indeed minimized when $\bS$ takes this form. We first do so in the low and high SNR regimes separately before connecting these regimes via analytic continuation to give a proof of Theorem \ref{thrm1}.

The key idea behind the proofs for the low and high SNR regimes is to find lower bounds $\widetilde{F}^{\rm RS}_{M,p,{\rm low}}$, respectively $\widetilde{F}^{\rm RS}_{M,p,{\rm high}}$, of $\widetilde{F}^{\rm RS}_{M,p}(\bS,\lambda)$ that are more clearly minimized by diagonal $\bS$, but are tight enough bounds that they agree with $\widetilde{F}^{\rm RS}_{M,p}(\bS,\lambda)$ at their minimizers.

\begin{proposition} \label{prop2}
Assume hypotheses \ref{H1}, \ref{H2}. Then, there exists $0<\lambda_S<\lambda_c$ such that whenever $0\le\lambda<\lambda_S$, we have
\begin{equation} \label{eq25}
\sup_{\bQ\in\cS_M}\widehat{F}^{\rm RS}_{M,p}(\bQ,\lambda)=\sup_{q\in[0,\rho^{p-1}]}\widehat{F}^{\rm RS}_{1,p}(q,\lambda).
\end{equation}
\end{proposition}
\begin{proof}
Let $\bQ^*\in\cS_M$ be a maximizer of $\widehat{F}^{\rm RS}_{M,p}(\bQ,\lambda)$ with eigenvalues $q_1^*,\ldots,q_M^*$. Then, by Lemma \ref{Lemma4}, there exists finite $\rho'$ such that
\begin{equation*}
0\le q_1^*,\ldots,q_M^*<\rho'
\end{equation*}
for all $\lambda\ge0$. Thus, the eigenvalues $\sigma_i=(\lambda q_i^*)^{-1}$ ($i\in[M]$) of the corresponding minimizer $\bS$ of $\widetilde{F}^{\rm RS}_{M,p}(\bS,\lambda)$ are such that
\begin{equation} \label{eq26}
\sigma_1,\ldots,\sigma_M>\frac{1}{\lambda\rho'}.
\end{equation}
Letting $\bO_i$ be the eigenvector of $\bS$ corresponding to $\sigma_i$ and setting $\lambda_S:=D^{-2}/\rho'$, we have for $0\le\lambda<\lambda_S$ and $i\in[M]$ the following lower bound on the diagonal entries of $\bS$:
\begin{equation*}
\bS_{ii}=\sum_{j=1}^M\bO_{ij}^2\sigma_j>\frac{1}{\lambda\rho'}\sum_{j=1}^M\bO_{ij}^2=\frac{1}{\lambda\rho'}>D^2.
\end{equation*}
Thus, we have by Lemma \ref{Lemma2} that the mutual information term of the potential \eqref{eq24} satisfies inequality \eqref{eq15}. To address the tracial term in said potential, we substitute $\bQ^*=(\lambda\bS)^{-1}$ into inequality \eqref{eq21} to see that
\begin{multline} \label{eq27}
\frac{p-1}{2p}\lambda^{-1/(p-1)}\sum_{\ell,\ell'=1}^M\bS_{\ell\ell'}^{-p/(p-1)}
\\ \ge\frac{p-1}{2p}(\lambda M)^{-1/(p-1)}(\Tr\bS^{-1})^{p/(p-1)}.
\end{multline}
Using the AM-HM inequality $\Tr\bS^{-1}\ge M/\sigma$ to further reduce this term (recalling that $\sigma:=\Tr\bS/M$) and combining the result with inequality \eqref{eq15} yields
\begin{align}
\inf_{\bS\in\cS_M}\widetilde{F}^{\rm RS}_{M,p}(\bS,\lambda)\ge&\inf_{\sigma_1,\ldots,\sigma_M\ge0}\widetilde{F}^{\rm RS}_{M,p,{\rm low}}(\bs,\lambda), \label{eq28}
\\\widetilde{F}^{\rm RS}_{M,p,{\rm low}}(\bs,\lambda):=&MI(x_0;x_0+\sqrt{\sigma}z)-\frac{\rho}{2}\sum_{i=1}^M\frac{1}{\sigma_i} \nonumber
\\&+\frac{M(p-1)}{2p}(\lambda\sigma^p)^{-1/(p-1)}, \nonumber
\end{align}
where we write $\bs=(\sigma_1,\ldots,\sigma_M)$ and recall that $x_0\sim\pP_X$ and $z\sim\mathcal{N}(0,1)$.

Now, to minimize this simplified potential, we define
\begin{align*}
\tau(\sigma):=&\frac{d}{d\sigma}\left\{I(x_0;x_0+\sqrt{\sigma}z)+\frac{p-1}{2p}(\lambda\sigma^p)^{-1/(p-1)}\right\}
\\=&-\frac{1}{2\sigma^2}\mathrm{mmse}\left(\frac{1}{\sigma}\right)-\frac{1}{2\sigma}(\lambda\sigma^p)^{-1/(p-1)},
\end{align*}
where the second equality follows from the I-MMSE relation for scalar Gaussian channels \cite{guoshamaiverdu}. Then, minimizers of $\widetilde{F}^{\rm RS}_{M,p,{\rm low}}(\bs,\lambda)$ that are critical points satisfy, for all $i\in[M]$,
\begin{equation} \label{eq29}
0=\frac{\partial}{\partial\sigma_i}\widetilde{F}^{\rm RS}_{M,p,{\rm low}}(\bs,\lambda)=\tau(\sigma)+\frac{\rho}{2\sigma_i^2}.
\end{equation}
Consequently, we must have
\begin{equation*}
\sigma_i=\sqrt{-\frac{\rho}{2\tau(\sigma)}}
\end{equation*}
for all $i\in[M]$. As the right-hand side is independent of $i$, this implies that $\sigma_1=\cdots=\sigma_M$, so that $\bS=\sigma I_M$. 

It remains to consider minimizers $\bs$ on the boundary of $[0,\infty)^M$. For $\bs$ to be such a minimizer, one of the eigenvalues must be zero or tend to infinity. The former case is not possible due to the bound \eqref{eq26}. In the latter case, we must have that $\sigma\to\infty$, hence $\tau(\sigma)\to0$ due to the boundedness of the MMSE (see, e.g., \cite{LelargeMiolane}). Now, suppose for the sake of contradiction that there exists $j\in[M]$ such that the eigenvalue $\sigma_j$ is finite. Then, by the right-hand side of equation \eqref{eq29}, we have
\begin{equation*}
\frac{\partial}{\partial\sigma_j}\widetilde{F}^{\rm RS}_{M,p,{\rm low}}(\bs,\lambda)\to\frac{\rho}{2\sigma_j^2},
\end{equation*}
which is strictly positive for finite $\sigma_j$, so such a $\bs$ cannot be a minimizer. Thus, if $\bs$ is a minimizer of $\widetilde{F}^{\rm RS}_{M,p,{\rm low}}(\bs,\lambda)$ with an eigenvalue tending to infinity, then all of the eigenvalues must tend to infinity, and we may again write $\bS=\sigma I_M$.

We have established that a minimizer $\bs$ of $\widetilde{F}^{\rm RS}_{M,p,{\rm low}}(\bs,\lambda)$, regardless of whether it is a critical or boundary point, must have equal components when $0\le\lambda<\lambda_S$. We thus see the decoupling
\begin{align*}
&\inf_{\sigma_1,\ldots,\sigma_M\ge0}\widetilde{F}^{\rm RS}_{M,p,{\rm low}}(\bs,\lambda)
\\&=M\inf_{\sigma\ge0}\Big\{I(x_0;x_0+\sqrt{\sigma}z)
\\&\hspace{7em}-\frac{\rho}{2\sigma}+\frac{p-1}{2p}(\lambda\sigma^p)^{-1/(p-1)}\Big\}.
\end{align*}
Moreover, if $\sigma^*$ is a minimizer of this problem, $\widetilde{F}^{\rm RS}_{M,p}(\sigma^*I_M,\lambda)$ equals the above, and inequality \eqref{eq28} is in fact an equality. Thus, the minimizers of $\widetilde{F}^{\rm RS}_{M,p}(\bS,\lambda)$ are of the form $\bS=\sigma I_M$ and, consequently, the maximizers of $\widehat{F}^{\rm RS}_{M,p}(\bQ,\lambda)$ are of the form $\bQ=q I_M$. Substituting this into the left-hand side of equation \eqref{eq25} produces the right-hand side, as desired.
\end{proof}

\begin{proposition} \label{prop3}
Assume hypotheses \ref{H1}--\ref{H3}. Then, there exist $\lambda_L>\lambda_c$ such that whenever $\lambda>\lambda_L$,
\begin{equation} \label{eq30}
\sup_{\bQ\in\cS_M}\widehat{F}^{\rm RS}_{M,p}(\bQ,\lambda)=\sup_{q\in[0,\rho^{p-1}]}\widehat{F}^{\rm RS}_{1,p}(q,\lambda).
\end{equation}
\end{proposition}
\begin{proof}
We once again minimize $\widetilde{F}^{\rm RS}_{M,p}(\bS,\lambda)$ over $\bS\in\cS_M$. In the high SNR setting, the bound \eqref{eq26} can no longer be used to ensure that the diagonal entries $\bS_{ii}$ are large enough for inequality \eqref{eq15} to apply. Instead, we know from Lemma~\ref{Lemma6} that there exists $\lambda_L'>\lambda_c$ such that whenever $\lambda>\lambda_L'$, the eigenvalues $\sigma_1,\ldots,\sigma_M$ of $\bS=(\lambda\bQ^*)^{-1}$ are such that
\begin{equation} \label{eq31}
\frac{1}{\lambda\rho^{p-1}}\le\sigma_1,\ldots,\sigma_M<\frac{1}{\lambda\rho_L'},
\end{equation}
where we now set
\begin{equation} \label{eq32}
\rho_L':=\left(\tfrac{2(p-1)}{2p-1}\rho\right)^{p-1}\in(0,\rho^{p-1}).
\end{equation}
Before utilizing these bounds, we must first apply a series of inequalities to reduce the tracial term of $\widetilde{F}^{\rm RS}_{M,p}(\bS,\lambda)$ seen in equation \eqref{eq24}.

For ease of notation, we actually begin with the tracial term of $\widehat{F}^{\rm RS}_{M,p}(\bQ,\lambda)$. Thus, focusing first on the third term displayed in equation \eqref{eq11}, we observe that
\begin{align}
\Tr(\bQ\bA)&=\sum_{\ell,\ell'=1}^M\bQ_{\ell\ell'}^{p/(p-1)} \nonumber
\\&\ge\sum_{\ell=1}^M\left(\sum_{\ell'=1}^M\bQ_{\ell\ell'}^2\right)^{p/(2p-2)} \nonumber
\\&=\sum_{\ell=1}^M(\bQ^2)_{\ell\ell}^{p/(2p-2)} \nonumber
\\&=\sum_{\ell=1}^M\left(\sum_{\ell'=1}^M\bO_{\ell\ell'}^2q_{\ell'}^2\right)^{p/(2p-2)} \nonumber
\\&\ge\sum_{\ell,\ell'=1}^M\bO_{\ell\ell'}^2q_{\ell'}^{p/(p-1)} \nonumber
\\&=\sum_{i=1}^Mq_i^{p/(p-1)}, \label{eq33}
\end{align}
where $q_1,\ldots,q_M$ are the eigenvalues of $\bQ$ with corresponding eigenvectors $\bO_1,\ldots,\bO_M$. We justify each line of the above as follows:
\begin{enumerate}
\item Read off from Lemma \ref{Lemma1}.
\item First note that $R_i:=\bQ_{\ell i}^2/\sum_{\ell,\ell'=1}^M\bQ_{\ell\ell'}^2\le1$. Then, since $p/(2p-2)<1$, we have that $R_i\le R_i^{p/(2p-2)}$. Summing both sides of this inequality over $i=1,\ldots,M$ shows that $1\le\sum_{i=1}^MR_i^{p/(2p-2)}$ and we are done upon rearrangement of this result.
\item As $\bQ$ is symmetric, $\bQ^2=\bQ^\intercal\bQ$, the $\ell$-th diagonal entry of which is the sum $\sum_{\ell'=1}^M\bQ_{\ell\ell'}^2$.
\item Express $(\bQ^2)_{\ell\ell}$ as a sum in terms of its eigenvalues $q_{\ell'}^2$ and corresponding eigenvectors $\bO_{\ell\ell'}$ ($\ell'\in[M]$).
\item Apply Jensen's inequality using the fact that the function $x\mapsto x^{p/(2p-2)}$ is concave and $\sum_{\ell'=1}^M\bO_{\ell\ell'}^2=1$ due to $\bO_{\ell}$ being a unit vector.
\item Sum over $\ell$ and replace $\ell'$ by $i$, using yet again the fact that $\bO_{\ell'}$ is a unit vector.
\end{enumerate}

Now, substituting $\bQ=\lambda^{-1}\bS^{-1}$ into inequality \eqref{eq33}, we have that the tracial term in equation \eqref{eq24} satisfies
\begin{multline} \label{eq34}
-\frac{\rho}{2}\Tr\bS^{-1}+\frac{p-1}{2p}\lambda^{-1/(p-1)}\sum_{\ell,\ell'=1}^M(\bS^{-1})^{p/(p-1)}_{\ell\ell'}
\\\ge\sum_{i=1}^M\nu(\sigma_i),
\end{multline}
where we define
\begin{equation*}
\nu(\sigma):=-\frac{\rho}{2}\sigma^{-1}+\frac{p-1}{2p}\lambda^{-1/(p-1)}\sigma^{-p/(p-1)}.
\end{equation*}
The first and second order derivatives of this function are
\begin{align*}
\nu'(\sigma)&=\frac{\rho}{2}\sigma^{-2}-\frac{1}{2}\lambda^{-1/(p-1)}\sigma^{(1-2p)/(p-1)},
\\ \nu''(\sigma)&=-\rho\sigma^{-3}+\frac{2p-1}{2(p-1)}\lambda^{-1/(p-1)}\sigma^{(2-3p)/(p-1)}
\\&=-\sigma^{-3}\left(\rho-\frac{2p-1}{2(p-1)}(\lambda\sigma)^{-1/(p-1)}\right).
\end{align*}
Hence, recalling the definition of $\rho_L'$ in equation \eqref{eq32}, $\nu(\sigma)$ is convex on the interval $[0,1/(\lambda\rho_L')]$. As the domain \eqref{eq31} lies within this interval, $\nu(\sigma_i)$ is therefore a convex function of $\sigma_i$ for all $i\in[M]$ whenever $\lambda>\lambda_L'$. Hence, for such $\lambda$, we may apply Jensen's inequality to inequality \eqref{eq34} to show that
\begin{equation*}
-\frac{\rho}{2}\Tr\bS^{-1}+\frac{p-1}{2p}\lambda^{-1/(p-1)}\sum_{\ell,\ell'=1}^M(\bS^{-1})^{p/(p-1)}_{\ell\ell'}\ge M\nu(\sigma),
\end{equation*}
where we recall that $\sigma=\Tr\bS/M$. Combining this with inequality \eqref{eq14} of Lemma \ref{Lemma2} finally yields the lower bound
\begin{equation} \label{eq35}
\inf_{\bS\in\cS_M}\widetilde{F}^{\rm RS}_{M,p}(\bS,\lambda)\ge\inf_{\bS_{11},\ldots,\bS_{MM}\ge0}\widetilde{F}^{\rm RS}_{M,p,{\rm high}}(\diag\bS,\lambda),
\end{equation}
where we define the high SNR analog of $\widetilde{F}^{\rm RS}_{M,p,{\rm low}}$ as
\begin{equation*}
\widetilde{F}^{\rm RS}_{M,p,{\rm high}}(\diag\bS,\lambda):=\sum_{i=1}^MI(x_0;x_0+\sqrt{\bS_{ii}}z)+M\nu(\sigma).
\end{equation*}

As in the proof of Proposition \ref{prop2}, we are now tasked with minimizing the simplified potential defined above. Thus, we first consider critical points of $\widetilde{F}^{\rm RS}_{M,p,{\rm high}}$, which are solutions to the system of equations
\begin{align}
0&=\frac{\partial}{\partial\bS_{ii}}\widetilde{F}^{\rm RS}_{M,p,{\rm high}}(\diag\bS,\lambda),\quad i=1,\ldots,M \nonumber
\\&=\frac{d}{d\bS_{ii}}I(x_0;x_0+\sqrt{\bS_{ii}}z)+\nu'(\sigma). \label{eq36}
\end{align}
By the I-MMSE relation for scalar Gaussian channels \cite{guoshamaiverdu},
\begin{equation} \label{eq37}
\iota(\bS_{ii}):=\frac{d}{d\bS_{ii}}I(x_0;x_0+\sqrt{\bS_{ii}}z)=-\frac{1}{2}\bS_{ii}^{-2}\mathrm{mmse}(\bS_{ii}^{-1}),
\end{equation}
where $\mathrm{mmse}(t)$ is the MMSE of the scalar Gaussian channel $y=\sqrt{t}x_0+z$. We prove in Appendices \ref{A1}--\ref{A3} that for each of the cases in hypothesis \ref{H3}, there exists $\lambda_L''>0$ such that $t^2\mathrm{mmse}(t)$ is monotone for all $t>\lambda_L''$. Hence, setting $\lambda_L:=\mathrm{max}\{\lambda_L',\lambda_L''/\rho_L'\}$, we have that $\iota(\bS_{ii})$ is monotone when $\bS_{ii}<1/(\lambda_L\rho_L')$. We see that this is exactly the case for all $\lambda>\lambda_L$ since by the bound \eqref{eq31}, we have
\begin{equation*}
\bS_{ii}=\sum_{k=1}^M\bO_{ij}^2\sigma_j<\frac{1}{\lambda\rho_L'}\sum_{j=1}^M\bO_{ij}^2<\frac{1}{\lambda_L\rho_L'},
\end{equation*}
where $\bO_i$ is the eigenvector of $\bS$ corresponding to $\sigma_i$. Returning to the system of equations \eqref{eq36}, it is clear that when $\lambda>\lambda_L$, critical points of $\widetilde{F}^{\rm RS}_{M,p,{\rm high}}(\diag\bS,\lambda)$ satisfy
\begin{equation*}
\bS_{ii}=\iota^{-1}(\nu'(\sigma)),\quad i=1,\ldots,M.
\end{equation*}
As the right-hand side of this is independent of $i$, we must have $\bS_{11}=\ldots=\bS_{MM}$.

Consider now minimizers that are on the boundary of $[0,\infty)^M$. Such minimizers are not possible for finite $\lambda$ due to the bound \eqref{eq36}. When $\lambda\to\infty$, this same bound forces $\bS_{11},\ldots,\bS_{MM}\to0$, so we may once again write $\bS=\sigma I_M$. Finally, combining this observation with the above discussion on critical points, we conclude as in the proof of Proposition~\ref{prop2} that the right-hand side of inequality~\eqref{eq35} reduces as
\begin{multline*}
\inf_{\bS_{11},\ldots,\bS_{MM}\ge0}\widetilde{F}^{\rm RS}_{M,p,{\rm high}}(\diag\bS,\lambda)
\\ =M\inf_{\sigma\ge0}\left\{I(x_0;x_0+\sqrt{\sigma}z)+\nu(\sigma)\right\}.
\end{multline*}
As substituting $\bS=\sigma I_M$ into the left-hand side of said inequality produces the same, the remainder of this proof is as the conclusion of the proof of Proposition \ref{prop2}.
\end{proof}

In the above two proofs, we used Jensen's inequality and convexity to express either the mutual information term or the tracial term in terms of the scalar variable $\sigma=\Tr\bS/M$ --- a caveat is that in the first case, we actually grouped the $\Tr(\bQ\bS)$ term with the mutual information to produce a simpler proof than in \cite{BKR26}. In each of these cases, the remaining term is expressed in terms of the eigenvalues, respectively diagonal entries, of $\bS$. Due to monotonicity properties, these terms have invertible derivatives, forcing all eigenvalues, respectively diagonal entries, to be equal whenever $\bS$ is a minimizer of the potential at hand. Concavity and monotonicity properties are forced by taking $\lambda$ either small or large enough. However, these are only partial proofs of Theorem \ref{thrm1}.

With propositions \ref{prop2} and \ref{prop3} in hand, we now proceed with the proof of Theorem \ref{thrm1}.
\begin{proof}[Proof of Theorem \ref{thrm1}]
By propositions \ref{prop2} and \ref{prop3}, we have that
\begin{equation*}
\sup_{\bQ\in\cS_M}\widehat{F}^{\rm RS}_{M,p}(\bQ,\lambda)=\sup_{q\in[0,\rho^{p-1}]}\widehat{F}^{\rm RS}_{1,p}(q,\lambda)
\end{equation*}
for all $\lambda\in[0,\lambda_S)\cup(\lambda_L,\infty)$, where $0<\lambda_S<\lambda_c<\lambda_L$ are as in said propositions. Moreover, these suprema are obtained when $\bQ$ is a scalar multiple of the identity, so we may replace the supremum with a supremum over $\bQ\in\psi(\cS_M)$. Thus, by equation \eqref{eq10}, we have for $\lambda$ in the same domain,
\begin{align*}
\sup_{\bA\in\cS_M}F^{\rm RS}_{M,p}(\bA,\lambda)&=\sup_{q\in[0,\rho^{p-1}]}\widehat{F}^{\rm RS}_{1,p}(q,\lambda)
\\&=\sup_{s\in[0,\rho]}F^{\rm RS}_{1,p}(s,\lambda).
\end{align*}
The left-hand and right-hand sides of this equation are respectively $\phi_M(\lambda)$ and $\phi_1(\lambda)$. By hypothesis \ref{H4}, both of these are real analytic on $[0,\infty)\setminus\{\lambda_c\}$, so by the identity theorem for real analytic functions, the agreement between $\phi_1(\lambda)$ and $\phi_M(\lambda)$ on $[0,\lambda_S)\cup(\lambda_L,\infty)$ described above extends to all of $[0,\infty)\setminus\{\lambda_c\}$. Since $\phi_1(\lambda)$ and $\phi_M(\lambda)$ are also continuous functions of $\lambda$ (see Remark 2.2 of \cite{BKR26}), they must agree at $\lambda=\lambda_c$, as well.
\end{proof}

\section{Conclusion} \label{s5}
We have shown that when the entries of a ground truth signal matrix $\bX_0$ are i.i.d. and satisfy some mild constraints, the limiting free entropy associated to the rank-$M$, order-$p$, symmetric, spiked tensor factorization problem is given by the same formula as when $M=1$. From an information-theoretic viewpoint, we have thus shown that the limiting mutual information is given by (see, e.g., \cite{LelargeMiolane})
\begin{equation*}
\lim_{N\to\infty}\frac{1}{NM}I(\bX_0;\bY)=\frac{\rho^p\lambda}{2p}-\sup_{s\in[0,\rho]}F^{\rm RS}_{1,p}(s,\lambda).
\end{equation*}
Furthermore, we may read from the rank-one formula given in Theorem 2 of \cite{spikedtensor} that the tensor-MMSE has large $N$ limit
\begin{equation*}
\lim_{N\to\infty}\textrm{T-MMSE}_N(\lambda)=\rho^p-s^*(\lambda)^p,\quad\lambda\ne\lambda_c,
\end{equation*}
where $s^*(\lambda)$ is the unique maximizer of $F^{\rm RS}_{1,p}(s,\lambda)$. At a high level, this means that if one wishes to theoretically model complicated tensorial data, the model must break at least one of the hypotheses of Theorem \ref{thrm1} and one is encouraged to develop more structured models.

One of the key motivations of this work is that it is a natural generalisation of the matrix case studied in \cite{BKR26}. There, a multiscale cavity method was introduced, which allowed for the treatment of the spiked Wigner model with rank growing slowly with $N$. A key ingredient in the proof was that the finite-rank case reduces to the rank-one result. Thus, in this work, we have taken the first step in replicating the results of \cite{BKR26} for the tensor case. Along the way, we have also been able to slightly simplify and strengthen the proofs and improve a hypothesis (replacing the hypothesis on monotonicity of $\lambda^2\mathrm{mmse}(\lambda)$ with the requirement that $\pP_X$ have no continuous singular part). In principle, the multiscale cavity method can be used to extend the regime of veracity for Theorem \ref{thrm1} to some rank $M$ growing slowly in $N$ --- based on preliminary calculations stemming from \cite{spikedtensor}, we loosely conjecture $M={\rm o}((\ln N)^{1/p})$. However, a full computation is outside the scope of the present work. 

\vspace{2ex}

\section*{Acknowledgments}
The work of AAR is supported by Hong Kong RGC grants GRF 16304724 and GRF 17304225. RHM initiated this work as a master's student at The University of Trieste and is presently a PhD candidate at The University of Turin. The authors are grateful to Jean Barbier and Justin Ko for in depth discussions on this project.

\appendices

\section{Monotonicity of \texorpdfstring{$\iota(t)$}{iota(t)} at small \texorpdfstring{$t$}{t} for discrete prior distributions} \label{A1}
We consider the scalar Gaussian channel $y=\sqrt{t}x_0+z$ where $z\sim\mathcal{N}(0,1)$ and $x_0\sim\pP_X$ with $\pP_X$ being a discrete distribution with $D$-bounded support. Then, there exists finite $K\in\mathbb{N}$ and $-D\le a_1<\ldots<a_K\le D$ such that
\begin{equation*}
\pP_X(x)=\sum_{i=1}^Kp_i\delta(x-a_i),\quad p_i:=\pP(x_0=a_i).
\end{equation*}
By Bayes' rule, the conditional expectation is given by
\begin{align}
\E[\,\cdot\mid y]&=\frac{\int_{\R}\,\cdot\,e^{-(y-\sqrt{t}x)^2/2}\,d\pP_X(x)}{\int_{\R}e^{-(y-\sqrt{t}x)^2/2}\,d\pP_X(x)} \label{eqA0}
\\&=\frac{\sum_{i=1}^K\,\cdot\,e^{-(y-\sqrt{t}a_i)^2/2}p_i}{\sum_{i=1}^Ke^{-(y-\sqrt{t}a_i)^2/2}p_i}. \nonumber
\end{align}
Adapting the strategy of \cite{alvarado}, we then see that the conditional variance is given by
\begin{equation*}
\Var[x_0\mid y]=\frac{\sum_{i,j=1}^K(a_i-a_j)^2e^{-\tfrac{t}{2}(u-a_i)^2-\tfrac{t}{2}(u-a_j)^2}p_ip_j}{2\sum_{i,j=1}^Ke^{-\tfrac{t}{2}(u-a_i)^2-\tfrac{t}{2}(u-a_j)^2}p_ip_j},
\end{equation*}
where we have made the change of variables $u=y/\sqrt{t}$. We now wish to analyse the large $t$ behaviour of the above expression. To do so, we define for each $u\in\R$,
\begin{align*}
a_\alpha&:=\argmin_{\{a_1,\ldots,a_K\}}\{(u-a_i)^2\},
\\a_\beta&:=\argmin_{\{a_1,\ldots,a_K\}\setminus\{a_\alpha\}}\{(u-a_i)^2\}
\end{align*}
to be the closest and second closest points of $u$ in $\supp\pP_X$. At large $t$, we may discard negligible terms to write
\begin{align*}
\Var[x_0\mid y]=&(a_\alpha-a_\beta)^2\frac{\varepsilon(u)}{(\varepsilon(u)+1)^2}+{\rm o}(e^{-\tfrac{t}{2}(a_\alpha-a_\beta)^2}),
\\\varepsilon(u):=&\frac{p_\beta}{p_\alpha}e^{\tfrac{t}{2}(u-a_\alpha)^2-\tfrac{t}{2}(u-a_\beta)^2}.
\end{align*}
Observe that the above formulas stay consistent (otherwise, $a_\alpha,a_\beta$ change meaning) when $(u-a_\alpha)^2<(u-a_i)^2$ for all $i\ne\alpha$. Thus, partitioning $\R=\cup_{i=1}^K\mathcal{I}_i$ with
\begin{align*}
\mathcal{I}_1&:=(-\infty,\tfrac{a_1+a_2}{2}),
\\ \mathcal{I}_j&:=(\tfrac{a_{j-1}+a_j}{2},\tfrac{a_j,a_{j+1}}{2}),\quad j=2,\ldots,K-1,
\\ \mathcal{I}_K&:=(\tfrac{a_{K-1}+a_K}{2},\infty),
\end{align*}
we have for $n\in\N$ that
\begin{multline} \label{eqA1}
\E_{z,x_0}\left[\Var[x_0\mid y]^n\right]=\sqrt{\tfrac{t}{2\pi}}\sum_{i=1}^Kp_i\delta_i^n(1+{\rm o}(e^{-\tfrac{t}{2}\delta_i}))
\\\times\int_{\mathcal{I}_i}\frac{\varepsilon(u)^n}{(\varepsilon(u)+1)^{2n}}e^{-\tfrac{t}{2}(u-a_i)^2}\,du,
\end{multline}
where $\delta_i=\min_{j\ne i}\{(a_i-a_j)^2\}$. For each integral, changing variables to $v=t\sqrt{\delta_i}(u+a_i)$ and using Laplace's method yields
\begin{multline*}
\int_{\mathcal{I}_i}\frac{\varepsilon(u)^n}{(\varepsilon(u)+1)^{2n}}e^{-\tfrac{t}{2}(u-a_i)^2}\,du
\\=\frac{2^{1-2n}}{t\sqrt{\delta_i}}e^{-\tfrac{\delta_it}{8}}(1+{\rm O}(t^{-1}))\int_{-\infty}^0\mathrm{sech}^{2n}(v)e^v\,dv.
\end{multline*}
Substituting this into equation \eqref{eqA1} shows that there exist finite constants $C_n$ and $\delta^*=\min_i\{\delta_i\}$ such that for each $n\in\N$,
\begin{equation} \label{eqA2}
\E_{z,x_0}\left[\Var[x_0\mid y]^n\right]=\left(C_nt^{-1/2}+{\rm O}(t^{-1})\right)e^{-\delta^*t/8}.
\end{equation}

We are finally ready to address $\iota(t)$ as defined in \eqref{eq37}. For simplicity, we first consider $\iota(t^{-1})$, which has derivative
\begin{align*}
\frac{d}{dt}\iota(t^{-1})&=-t\mathrm{mmse}(t)-\tfrac{1}{2}t^2\mathrm{mmse}'(t)
\\&=-t\E_{z,x_0}\left[\Var[x_0\mid y]\right]+\tfrac{1}{2}t^2\E_{z,x_0}\left[\Var[x_0\mid y]^2\right]
\\&=(-C_1t^{1/2}+\tfrac{1}{2}C_2t^{3/2} +{\rm O}(t))e^{-\delta^*t/8},
\end{align*}
with the second line due to Proposition 9 of \cite{analyticMMSE} and the third following from equation \eqref{eqA2}. It is clear that this is strictly positive for large enough $t$, so we may define some $\lambda_L''>0$ such that $\iota(t^{-1})$ is monotonic for all $t>\lambda_L''$. Then, $\iota(t)$ is monotonic for all $t<1/\lambda_L''$.

\section{Monotonicity of \texorpdfstring{$\iota(t)$}{iota(t)} at small \texorpdfstring{$t$}{t} for absolutely continuous prior distributions} \label{A2}
We repeat the exercise of Appendix \ref{A1}, but with $\pP_X$ now an absolutely continuous distribution with $D$-bounded support. Once again, our first object of interest is the conditional variance for the channel $y=\sqrt{t}x_0+z$,
\begin{align*}
\Var[x_0\mid y]&=\E[x_0^2\mid y]-\E[x_0\mid y]^2
\\&=\tfrac{1}{t}\left(\E[w^2\mid u]-\E[w\mid u]^2\right),
\end{align*}
where the conditional expectation is as in equation \eqref{eqA0} and we have changed variables to $w=\sqrt{t}(x_0-u)$ with $u=y/\sqrt{t}$. Since $\pP_X$ is absolutely continuous, $\lim_{t\to\infty}\pP_X(u+w/\sqrt{t})=\pP_X(u)$ and, moreover, $w^ne^{-w^2/2}\pP_X(u+w/\sqrt{t})$ is dominated by $Mw^ne^{-w^2/2}$ for some finite constant $M$. Thus, we may use the dominated convergence theorem to see that at large $t$,
\begin{align*}
\E[w\mid u]&={\rm o}(1),
\\ \E[w^2\mid u]&=1+{\rm o}(1),
\end{align*}
and consequently,
\begin{equation} \label{eqB0}
\Var[x_0\mid y]=t^{-1}+{\rm o}(t^{-1}).
\end{equation}
Furthermore, since $\supp\pP_X$ is $D$-bounded, we have the uniform bound $\Var[x_0\mid y]<D^2$ and hence, there exist constants $C_n$ such that for each $n\in\N$,
\begin{equation*}
t^n\Var[x_0\mid y]^n\le C_n.
\end{equation*}
This follows from using equation \eqref{eqB0} and splitting the domain of $t$ as $[0,1]\cup[1,\infty)$. Thus, $\frac{1}{\sqrt{2\pi}}t^n\Var[x_0\mid y]^ne^{-z^2/2}\pP_X(x_0)$ is dominated by $\frac{1}{\sqrt{2\pi}}C_ne^{-z^2/2}\pP_X(x_0)$ and we have by the dominated convergence theorem that
\begin{equation*}
\lim_{t\to\infty}t^n\E_{z,x_0}[\Var[x_0\mid y]^n]=\E_{z,x_0}[\lim_{t\to\infty}t^n\Var[x_0\mid y]^n]=1.
\end{equation*}
As in the discrete case of Appendix \ref{A1}, we finally have
\begin{align*}
\frac{d}{dt}\iota(t^{-1})&=-t\mathrm{mmse}(t)-\tfrac{1}{2}t^2\mathrm{mmse}'(t)
\\&=-t\E_{z,x_0}\left[\Var[x_0\mid y]\right]+\tfrac{1}{2}t^2\E_{z,x_0}\left[\Var[x_0\mid y]^2\right]
\\&\to1/2.
\end{align*}
Thus, we may define some $\lambda_L''>0$ such that $\iota(t^{-1})$, respectively $\iota(t)$, is monotonic for all $t>\lambda_L''$, respectively $t<1/\lambda_L''$.

\section{Monotonicity of \texorpdfstring{$\iota(t)$}{iota(t)} at small \texorpdfstring{$t$}{t} for mixtures of discrete and absolutely continuous distributions} \label{A3}
We now combine the results of the previous two appendices to treat the case where $\pP_X$ is a mixture of a discrete and an absolutely continuous distribution. Thus, suppose that 
\begin{equation*}
\pP_X=\gamma\pP_0+(1-\gamma)\pP_1,
\end{equation*}
where $0<\gamma<1$ and $\pP_0,\pP_1$ are as in Appendix \ref{A1} and Appendix \ref{A2}, respectively. By the Lebesgue decomposition theorem \cite{cinlar}, these distributions are mutually singular, so we may define the boolean event
\begin{equation*}
R=\begin{cases}0,&x_0\in\supp\pP_0,\\1,&x_0\in\supp\pP_1.\end{cases}
\end{equation*}
Then, we have $\pP(R=0)=\gamma$ and $\pP(R=1)=1-\gamma$ and the corresponding posterior weights for $i=0,1$ are
\begin{align*}
\pi_i(y)=&\E[\1_{R=i}\mid y]=\gamma^{i-1}(1-\gamma)^i\frac{f_i(y)}{f_Y(y)},
\\ f_i(y):=&\frac{1}{\sqrt{2\pi}}\int_\R e^{-(y-\sqrt{t}x)^2}\,d\pP_i(x),
\\ f_Y(y):=&\frac{1}{\sqrt{2\pi}}\int_\R e^{-(y-\sqrt{t}x)^2}\,d\pP_X(x).
\end{align*}
Now, by the law of total variance, we have that
\begin{equation} \label{eqC1}
\Var[x_0\mid y]=\E_R[\Var[x_0\mid y,R]]+\Var_R[\E[x_0\mid y,R]],
\end{equation}
where
\begin{align*}
\E_R[\Var[x_0\mid y,R]]&=\sum_{i=0,1}\pi_i(y)\Var[x_0\mid y,R=i],
\\ \Var_R[\E[x_0\mid y,R]]&=\pi_0(y)\pi_1(y)\Big(\E[x_0\mid y,R=0]
\\&\hspace{7em}-\E[x_0\mid y,R=1]\Big)^2.
\end{align*}
Control of the first term follows from the results of the previous two appendices. For the second term, we must argue by partitioning the real line as in Appendix \ref{A1}.

Thus, let $\supp\pP_0=\{a_1,\ldots,a_K\}$, $v_i=y-\sqrt{t}a_i$, and $\mathcal{I}_1,\ldots,\mathcal{I}_K$ be as in Appendix \ref{A1}. Then, for $n\in\mathbb{N}$,
\begin{align*}
&\E_{z,x_0}\left[\Var_R[\E[x_0\mid y,R]]^n\right]
\\&=\sum_{i=1}^K\int_{\mathcal{I}_i}\,dv_i\,\left(\frac{(1-\gamma)\pi_0(y)f_1(y)}{f_Y(y)}\right)^nf_Y(y)
\\&\hspace{4em}\times\Big(\E[x_0\mid y,R=0]-\E[x_0\mid y,R=1]\Big)^{2n}.
\end{align*}
Using similar arguments to those given in the previous two appendices to simplify each integral, we have that as $t\to\infty$,
\begin{equation*}
\Big(\E[x_0\mid y,R=0]-\E[x_0\mid y,R=1]\Big)^{2n}=\left(\frac{v_i^2}{t}\right)^n(1+{\rm o}(1)),
\end{equation*}
along with
\begin{align*}
\pi_0(y)^n&\to1,
\\ f_1(y)^n&=e^{-nv_i^2/2}\times{\rm O}(t^{-n/2}),
\\ f_Y(y)^{1-n}&=e^{(n-1)v_i^2/2}\times{\rm O}(1).
\end{align*}
Combining all of this together shows that
\begin{align}
&\E_{z,x_0}\left[\Var_R[\E[x_0\mid y,R]]^n\right] \nonumber
\\&=\sum_{i=1}^K\int_{\mathcal{I}_i}(1-\gamma)^nv_i^{2n}e^{-v_i^2/2}\,dv_i \times{\rm O}(t^{-3n/2}) \nonumber
\\&={\rm O}(t^{-3n/2}). \label{eqC2}
\end{align}

We are finally ready to discuss $\iota(t)$. By equation \eqref{eqC1},
\begin{multline} \label{eqC3}
\E_{z,x_0}[\Var[x_0\mid y]^2]=\E_{z,x_0}[\E_R[\Var[x_0\mid y,R]]^2]
\\+2\E_{z,x_0}[\E_R[\Var[x_0\mid y,R]]\Var_R[\E[x_0\mid y,R]]]
\\+\E_{z,x_0}[\Var_R[\E[x_0\mid y,R]]^2].
\end{multline}
As $\pi_i(y)=\E[\1_{R=i}\mid y]$, equations \eqref{eqA2}, \eqref{eqB0} yield
\begin{align*}
&\E_{z,x_0}[\E_R[\Var[x_0\mid y,R]]^2]
\\&=\gamma C_2t^{-1/2}e^{-\delta^*t/8}+(1-\gamma)t^{-2}+{\rm o}(t^{-2})
\\&=(1-\gamma)t^{-2}+{\rm o}(t^{-2}).
\end{align*}
Combining this with equation \eqref{eqC2} with $n=2$ and using the Cauchy--Schwarz inequality shows that
\begin{equation*}
\E_{z,x_0}[\E_R[\Var[x_0\mid y,R]]\Var_R[\E[x_0\mid y,R]]]={\rm O}(t^{-5/2}).
\end{equation*}
Inserting these last two results, along with equation \eqref{eqC2} with $n=2$, into equation \eqref{eqC3} produces
\begin{equation*}
\E_{z,x_0}[\Var[x_0\mid y]^2]=(1-\gamma)t^{-2}+{\rm o}(t^{-2}).
\end{equation*}
It is known from \cite{wuverdu} that
\begin{equation*}
\E_{z,x_0}[\Var[x_0\mid y]]=(1-\gamma)t^{-1}+{\rm o}(t^{-1}),
\end{equation*}
so all together, recalling the expression for $\frac{d}{dt}\iota(t^{-1})$, we see
\begin{align*}
\frac{d}{dt}\iota(t^{-1})&=-t\mathrm{mmse}(t)-\tfrac{1}{2}t^2\mathrm{mmse}'(t)
\\&=-t\E_{z,x_0}\left[\Var[x_0\mid y]\right]+\tfrac{1}{2}t^2\E_{z,x_0}\left[\Var[x_0\mid y]^2\right]
\\&=\frac{\gamma-1}{2}+{\rm o}(1).
\end{align*}
The monotonicity of $\iota(t)$ for $t<\lambda_L''$ follows from the same reasoning as in appendices \ref{A1}, \ref{A2}.

{\scriptsize
\bibliographystyle{IEEEtran}
}

\end{document}